\documentclass{article}
\usepackage[a4paper]{geometry}

\usepackage{amssymb,amsmath,amsthm}
\usepackage{graphics,graphicx,color}
\usepackage{microtype,verbatim,dsfont}
\usepackage{xspace,enumerate,paralist}
\usepackage{subcaption}
\usepackage[pdfpagelabels,colorlinks,citecolor=blue,linkcolor=blue,urlcolor=blue]{hyperref}

\theoremstyle{plain}
\newtheorem{theorem}{Theorem}
\newtheorem{definition}{Definition}
\newtheorem{lemma}[theorem]{Lemma}
\newtheorem{corollary}[theorem]{Corollary}

\newtheorem{claim}{Claim}

\newcommand{\drawing}{disk-link drawing\xspace}
\newcommand{\drawings}{\drawing{s}\xspace}
\newcommand{\modulo}[1]{\ \mathrm{mod}\ #1}

\title{Grid Drawings of Graphs with \\Constant Edge-Vertex Resolution}

\author{Michael A. Bekos$^1$, Martin~Gronemann$^2$, Fabrizio~Montecchiani$^3$,\\D\"om\"ot\"or~P\'alv\"olgyi$^4$, Antonios~Symvonis$^5$, Leonidas~Theocharous$^5$
\\
\medskip
\\
\small$^1$Wilhelm-Schickhard-Institut f\"ur Informatik, Universit\"at T\"ubingen, Germany\\
\small\texttt{bekos@informatik.uni-tuebingen.de}
\\
\small$^2$Theoretical Computer Science, Osnabr\"uck University, Osnabr\"uck, Germany\\
\small\texttt{martin.gronemann@uni-osnabrueck.de}
\\
\small$^3$Department of Engineering, Universit\'a degli Studi di Perugia, Italy\\
\small\texttt{fabrizio.montecchiani@unipg.it}
\\
\small$^4$MTA-ELTE Lend\"ulet Combinatorial Geometry Research Group, Institute of Mathematics, \\
\small E\"otv\"os Lor\'and University (ELTE), Budapest, Hungary\\
\small\texttt{dom@cs.elte.hu}
\\
\small$^5$School of Applied Mathematical \& Physical Sciences, NTUA, Greece\\
\small\texttt{symvonis@math.ntua.gr}, \texttt{leonid\_97@hotmail.com}
}
\date{} 

\begin{document}

\maketitle

\begin{abstract}
We study the algorithmic problem of computing drawings of graphs in which $(i)$~each vertex is a disk with fixed radius $\rho$, $(ii)$~each edge is a straight-line segment connecting the centers of the two disks representing its end-vertices, $(iii)$~no two disks intersect, and $(iv)$~the distance between an edge segment and the center of a non-incident disk, called \emph{edge-vertex resolution}, is at least $\rho$. We call such drawings \emph{\drawings}. 

In this paper we focus on the case of constant edge-vertex resolution, namely $\rho=\frac{1}{2}$ (i.e., disks of unit diameter). We prove that star graphs, which trivially admit straight-line drawings in linear area, require quadratic area in any such \drawing. On the positive side, we present constructive techniques that yield improved upper bounds for the area requirements of \drawings for several (planar and nonplanar) graph classes, including bounded bandwidth, complete, and planar graphs. In particular, the presented bounds for complete and planar graphs are asymptotically tight.
\end{abstract}

\section{Introduction}
\label{sec:introduction}

A \emph{drawing} $\Gamma$ of a graph $G=(V,E)$ is a mapping of each vertex $v \in V$ to a  distinct point $p(v)$ on the plane and of each edge $(u,v) \in E$ to a Jordan arc with endpoints at  $p(u)$ and $p(v)$. When edges are drawn as straight-line segments, the corresponding drawings are referred to as \emph{straight-line  drawings}. 
\emph{Drawing algorithms} are used to generate the mapping of vertices and edges to points and Jordan arcs on the plane, respectively. The produced drawings follow   conventions, or \emph{drawing styles}, which dictate the characteristic features of the drawing, e.g., whether edges have to be drawn as a single (straight-line) segment or are allowed to have ``bends'',  whether vertex placement has to follow a pattern (e.g., drawn on a circle, or on several parallel lines as a hierarchy), etc.
A drawing  algorithm usually aims to optimize some characteristic attributes of the drawing, having as ultimate goal to produce aesthetically pleasing and useful representations, i.e., drawings that reveal properties of the underlying graphs and/or  facilitate their exploratory analysis. Drawing characteristics that we typically attempt to optimize include the number of edge crossings, the area of the drawing (assuming vertices at integer coordinates), the angular resolution and  the total number of bends (if bends are allowed). We refer the interested reader to references~\cite{DBLP:books/ph/BattistaETT99,DBLP:conf/dagstuhl/1999dg,DBLP:reference/crc/2013gd}.
 
Common to almost every drawing style are two restrictions that aim to eliminate any ambiguity on the drawn graph, and thus, to improve the readability of its drawing. They state that  \emph{edges cannot intersect (or pass over) vertices of the graph} and that \emph{edges cannot overlap each other}. Figure~\ref{fi:overlap_restrictions} demonstrates that when a drawing does not respect these restrictions we cannot interpret it in an unambiguous way. 

\begin{figure}[h]
	\centering
	\begin{subfigure}{0.3\textwidth}
	\includegraphics[width=\textwidth,page=1]{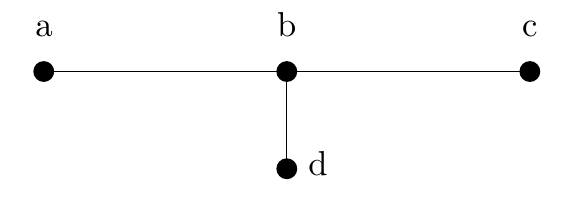}
	\subcaption{}
	\end{subfigure}
	\hfil
	\begin{subfigure}{0.3\textwidth}
	\includegraphics[width=\textwidth,page=2]{figs/intro}
	\subcaption{}
	\end{subfigure}
	\caption{%
	(a) A potential edge-vertex intersection. Does the graph  consist of two edges (i.e., $(a,c)$  and $(b,d)$) or three (i.e., $(a,b)$, $(b,c)$, $(b,d)$)?
	(b) A potential edge-edge overlap. Does the graph  consist of two edges (i.e., $(a,d)$  and $(b,c)$) or three (i.e., $(a,b)$, $(b,c)$, $(c,d)$)?}
	\label{fi:overlap_restrictions}
\end{figure}

When formalizing these restrictions for straight-line drawings, we require that  the line segment $\overline{p(u)p(v)}$ representing any edge $(u,v) \in E$ does not contain any point $p(w)$, where $w \in V$. Thus,  when trying to enforce these restrictions, edges are treated as line segments of zero width, and vertices as points. 
However, in reality, to easily identify the vertices we draw them as ``thick'' objects, and typically in the shape of a ``unit'' size disk or square.  In this scenario,  we have to make sure that no edge  intersects a non-incident vertex object. 

Reality dictates another restriction.   Graph drawings are usually displayed on a drawing canvas, where the centers of the vertex objects are  being placed at grid positions, i.e., they have integer coordinates.  So, when combined with the  requirement that vertices are of unit size, we are left with the following generic drawing problem: \emph{Given a graph $G$, produce a grid drawing $\Gamma$ of $G$ where the vertices are represented by unit-sized disks, the edges as (zero-width) line-segments and no edge intersects  any vertex~disk.}  Depending on the drawing style, the produced drawing may have to satisfy additional restrictions (e.g., no edge crossings). 

By assuming that our vertex objects are disks with unit diameter (our research can be extended to different objects with different size or shape), we call the grid drawings that have no overlaps between edges and vertices \emph{\drawings}; see Section~\ref{sec:basics} for a formal definition.
Graph editors typically create grid drawings with unit-sized disk vertices but they do not necessarily respect the ``no intersection between edge and vertex objects'' restriction; see Figure~\ref{fi:edge_disk_crossing-a}. As shown in Figure~\ref{fi:edge_disk_crossing-b},  by scaling the coordinates of the vertices up by a sufficiently large factor the problem is resolved, but at the cost of increasing the area of the drawing. We address precisely this problem: We design algorithms that compute \drawings in small area (smaller than the ones obtained by simple scaling up). 
Our research has several interesting connections to other problems studied in Graph Drawing and Computational Geometry, which we discuss below.

\begin{figure}[t]
	\centering
	\begin{subfigure}{0.2\textwidth}
	\includegraphics[width=\textwidth]{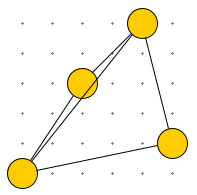}
	\subcaption{}
	\label{fi:edge_disk_crossing-a}
	\end{subfigure}
	\hfil
	\begin{subfigure}{0.2\textwidth}
	\includegraphics[width=\textwidth]{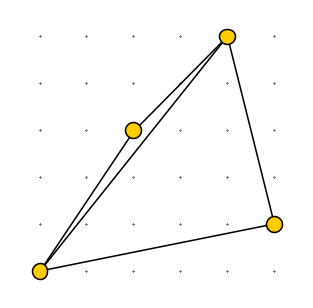}
	\subcaption{}
	\label{fi:edge_disk_crossing-b}
	\end{subfigure}
	\caption{%
	(a) Grid drawing of a graph created with the popular yEd graph editor (\url{https://www.yworks.com/yed-live/}): it contains an edge-vertex crossing. 
	(b) After scaling up (Tools $\rightarrow$ Geometric Transformations $\rightarrow$ Scale), the edge-vertex crossing is resolved.}
\label{fi:edge_disk_crossing}
\end{figure}

\paragraph{Related work} 
The problem of reducing the visual clutter caused by overlaps of vertices and edges has been recognized from the early years of Graph Drawing.  Davidson and Harel~\cite{DBLP:journals/tog/DavidsonH96}, back in 1996, presented a heuristic force-directed method to draw graphs nicely, which is based on an energy function having a term to penalize edges that are too close to vertices in the drawing.  Gansner and North~\cite{DBLP:conf/gd/GansnerN98} and Dobkin et al.~\cite{DBLP:conf/compgeom/DobkinHGN99} used two post-processing heuristics to improve drawings by reducing the visual clutter while conserving area. 

A \drawing can be equivalently considered as a traditional grid drawing in which the minimum distance between a vertex and an edge, called \emph{edge-vertex resolution}, is at least~$\frac{1}{2}$.
The first paper providing theoretical foundations for the problem of computing grid drawings with good edge-vertex resolution is by Chroback, Goodrich and Tamassia in 1996. In a preliminary work~\cite{DBLP:conf/compgeom/ChrobakGT96}, they present drawing algorithms for computing both 2D and 3D convex grid drawings with constant edge-vertex resolution. In relation to our work, they state that every 3-connected planar graph with $n$ vertices admits a convex grid drawing with constant edge-vertex resolution in $(3n-7 )\times (3n-7)/2$ area. However, several proofs, including the algorithm supporting this last result and its analysis, are missing from the paper and postponed to a full version that has not appeared yet. We further note that the edge-vertex resolution has also been considered in convex grid drawings with prescribed outer face~\cite{DBLP:journals/jgaa/ChambersEGL12} and in grid drawings with bends along the edges and optimal angular resolution~\cite{DBLP:journals/jal/GoodrichW00}.

Relevant to our work are also the \emph{rectangle-of-influence} (RI for short) drawings, which are planar straight-line drawings such that no vertex lies inside the axis-parallel rectangle defined by the two ends of every edge~\cite{DBLP:conf/gd/AlamdariB12,DBLP:journals/dm/BarriereH12,DBLP:conf/gd/BiedlBM99,DBLP:conf/cccg/BiedlLMV16,DBLP:journals/dcg/MiuraMN09,DBLP:journals/comgeo/SadasivamZ11}. (In particular, we refer to \emph{weak} RI drawings for which the absence of vertices in the axis-parallel rectangle defined by any two vertices does not imply the existence of an edge in the graph.) If all such rectangles are considered as open (resp.\ closed) sets, we call the corresponding drawings \emph{open} RI (resp.\ a \emph{closed} RI). While open RI drawings ensure some kind of distance between vertices and edges, they cannot be interpreted as \drawings because an edge may still intersect a disk whose center is on the boundary of its open rectangle. On the other hand, any closed RI drawing whose vertices are at integer coordinates can be seen as a \drawing. This implies that \drawings in quadratic area exist for all plane graphs with no filled $3$-cycle~\cite{DBLP:journals/dm/BarriereH12,DBLP:conf/gd/BiedlBM99,DBLP:journals/comgeo/SadasivamZ11}, hence including, for instance, trees and outerplane graphs. On the other hand, any plane graph with a filled $3$-cycle (e.g., $K_4$) does not admit a closed RI drawing~\cite{DBLP:conf/gd/BiedlBM99}.

Another related research direction considers drawings where vertices are objects with integer coordinates and the edges are fat segments. Barequet et al.~\cite{DBLP:journals/jgaa/BarequetGR04}, in an attempt to visualize weighted graphs, study drawings where the width of each edge is proportional to its weight and the width of each vertex is proportional to the sum of the weights of its incident edges. 
For edges with ``zero'' width, their drawings appear to be of similar style to the ones we consider in this paper, however one significant difference remains; in the drawings produced in~\cite{DBLP:journals/jgaa/BarequetGR04} \emph{the edges do not connect the centers of the incident vertex-disks} but rather simply enter these vertex-objects through varying angles.
Duncan et al.~\cite{DBLP:journals/ijfcs/DuncanEKW06} also use fat edges but, in contrast to~\cite{DBLP:journals/jgaa/BarequetGR04}, they do not compute a drawing from scratch but rather try to extend an existing one without modifying the area of the layout. 

Van Kreveld~\cite{DBLP:journals/comgeo/Kreveld11} introduced and studied \emph{bold drawings}, in which vertices are drawn as disks of radius $r$ and edges as rectangles of width $w$, where $r>w/2$. He concentrated on \emph{good bold drawings}, defined (informally) as bold drawings having all of its vertices and edges at least partially visible, in the sense that the area covered by overlapping edges is not sufficient to completely hide any vertex disk or edge-rectangle. In this regard, \drawing{s} can be seen as a special case of bold drawings in which $r=\frac{1}{2}-\varepsilon$ and $w=2\varepsilon$, for some sufficiently small $\varepsilon>0$. However, the research on bold drawings has mainly focused on finding feasible and strictly positive values of $r$ and $w$, rather than on area bounds for fixed values of $r$ and $w$ (which is the main question addressed in this paper). In particular, in~\cite{DBLP:journals/comgeo/Kreveld11}, it is shown that if a typical graph drawing (i.e., with point vertices and zero-width edges) is in non-degenerate position (i.e., no edge intersects a non-incident vertex and no three edges pass through a common point), then there exist positive values $r$ and $w$ that will turn it into a bold drawing. Van Kreveld also presented algorithms for (i) deciding whether for given $r$ and $w$ values a drawing is bold, and (ii) for maximizing $r$ and/or $w$ for a given drawing so that it is turned to a bold one. Later, Pach~\cite{DBLP:journals/jgaa/Pach15} answered one question posed by Van Kreveld in~\cite{DBLP:journals/comgeo/Kreveld11}. Namely, he showed that every graph admits a bold drawing in which the region occupied by the union of disks and rectangles representing the vertices and edges does not contain any disk of radius $r$ other than the ones representing the graph's vertices (i.e., no vertices can be hidden).  

When the input is a complete graph, our problem generalizes the classical \emph{no-three-in-line problem}~\cite{dudeney_1917,erdos}, which asks for the maximum number of points that can be placed on an $n \times n$ grid such that no three points are collinear.  In this regard, a result by Wood \cite{DBLP:journals/comgeo/Wood05} states that every $n$-vertex $k$-colorable graph has an $O(k) \times O(n)$ grid drawing such that no three points are collinear.
In our model, the collinearity requirement is strengthened by the no edge-disk intersection rule.

We finally remark that a straight-line drawing on an integer grid using only the horizontal, vertical and $\pm 1$ slopes for its edges can be transformed into a \drawing. Consequently, triconnected cubic planar graphs (except $K_4$) admit \drawings on grids of quadratic size~\cite{DBLP:journals/tcs/GiacomoLM18,DBLP:conf/wg/Kant92}. 

\paragraph{Contribution and paper organization} Based on the literature overview given in the previous paragraph, the problem of computing \drawings in compact area has not been tackled before (even though few results carry over from similar problems). Our contribution is as follows:


\begin{itemize}
\item We first give preliminaries and some basic results (Section~\ref{sec:basics}). In particular, we give an upper bound on the stretching factor that turns any grid drawing into a \drawing. This result immediately implies some area upper bounds for \drawings of certain~graph~classes. 

\item We then study improved area bounds for nonplanar graphs (Section~\ref{sec:nonplanar}). We show that bounded bandwidth graphs admit \drawings in linear area. The latter result is obtained by exploiting a construction of Erd\H{o}s~\cite{erdos} for the no-three-in-line problem. Moreover, as the main result of this section, we prove that every $n$-vertex complete graph has a convex \drawing in $O(n^4)$ area, which is asymptotically optimal. The upper bound is obtained by using the corners of a regular $n$-gon as an initial placement of the vertices, and by suitably perturbing the position of each vertex to enforce~integer~coordinates.

\item Afterwards, we turn our attention to crossing-free \drawings (Section~\ref{sec:planar}). Surprisingly, we can prove that $n$-vertex star graphs, which trivially admit straight-line drawings in $O(n)$ area, require $\Omega(n^2)$ area in any \drawing. This also implies that the previously mentioned quadratic area upper bound for trees and outerplanar graphs by Biedl et al.~\cite{DBLP:conf/gd/BiedlBM99} is asymptotically tight for both \drawings and closed RI drawings of such graphs (previously, a quadratic lower bound was known for the area of closed RI drawings of irreducible triangulations~\cite{DBLP:journals/comgeo/SadasivamZ11}). On the positive side, we present constructive techniques for all planar graphs that produce \drawings which asymptotically meet our area lower bound. Namely, we provide a detailed description of a linear-time algorithm that supports the claim in~\cite{DBLP:conf/compgeom/ChrobakGT96}, and computes planar (not necessarily convex) \drawings  in $(3n-7) \times \lceil (3n-7)/2 \rceil$ area for $n$-vertex planar graphs. This result is obtained by extending a central technique by de~Fraysseix, Pach and Pollack~\cite{DBLP:journals/combinatorica/FraysseixPP90}.

\end{itemize}

\section{Preliminaries and Basic Results}
\label{sec:basics}

\subsection{Basic Graph Drawing definitions.} 
We only consider \emph{simple graphs}, that is, graphs with neither self-loops nor parallel edges.  A \emph{drawing} of a graph $G$ is a mapping of the vertices of $G$ to distinct points of the plane, and of the edges of $G$ to Jordan arcs connecting their corresponding endpoints. We only consider \emph{simple drawings}, where any two edges intersect in at most one point, which is either a common endpoint or an interior point where the two edges properly cross. A drawing is \emph{planar} if no two edges intersect, except possibly at a common endpoint. A graph is \emph{planar} if it admits a planar drawing. A planar drawing partitions the plane into topologically connected regions, called \emph{faces}. The infinite region is called the \emph{outer face}; any other face is an \emph{inner face}. A \emph{planar embedding} of a planar graph is an equivalence class of topologically-equivalent (i.e., isotopic) planar drawings. A planar graph with a given planar embedding is a \emph{plane graph}. For a detailed description of major results in planar graph drawing refer to~\cite{DBLP:reference/crc/Vismara13}.

A drawing is \emph{straight-line} if the Jordan arc representing any edge is a straight-line segment. In what follows we only consider straight-line drawings. The \emph{slope} $s$ of a line $\ell$ is the angle that a horizontal line needs to be rotated counter-clockwise in order to make it overlap with $\ell$. The \emph{slope} of a segment is the slope of the supporting line containing it. A \emph{grid drawing} is a straight-line drawing whose vertices are at integer coordinates.  

\subsection{Disk-link drawings.}  A \drawing is formally defined as follows.

\begin{definition}
A \emph{\drawing} $\Gamma$ of a graph $G$ maps each vertex of $G$ to a distinct open disk with radius $\rho>0$ and each edge of $G$ to a (zero-width) straight-line segment connecting the centers of the two disks corresponding to its end-vertices, so that $(i)$~the center of each~disk has integer coordinates, $(ii)$~no two disks intersect, and $(iii)$~no disk is intersected by a non-incident~edge. 
\end{definition}  

\noindent 
We assume, for simplicity, that $\rho=\frac{1}{2}$, i.e., the disks have unit diameter. In accordance to traditional grid drawings, the \emph{edge-vertex resolution} of a \drawing $\Gamma$ is the minimum distance between the center of any disk and any non-incident edge, which is at least $\rho=\frac{1}{2}$ by our assumption.
We say that a graph admits a \drawing (resp.\ a grid drawing) on a grid of size $W \times H$ (or, equivalently, in area $W \times H$), if the minimum axis-aligned box containing it has side lengths $W-1$ and $H-1$. 
Note that our assumption $\rho=\frac{1}{2}$ is not restrictive, since \emph{our results carry over for any constant radius} up to some multiplicative constant factor for the area. We now introduce a basic property which we use to transform a grid drawing into a \drawing; the \emph{$x$-} and \emph{$y$-span} of an edge $(u,v)$ whose endpoints are $(x_u,y_u)$ and $(x_v,y_v)$ in a grid drawing are the quantities $\sigma_x(u,v)=|x_u-x_v|$ and $\sigma_y(u,v)=|y_u-y_v|$, respectively.

\begin{lemma}\label{le:stretch}
Let $\Gamma$ be a grid drawing of a graph $G$ and let $(u,v)$ be an edge of $\Gamma$ such that $\sigma_x(u,v)=X$ and $\sigma_y(u,v)=Y$. Let $\Gamma'$ be the drawing obtained by mapping each vertex $v$ with coordinates $(x_v,y_v)$ in $\Gamma$ to the point $(x_v\cdot \phi_X, y_v \cdot \phi_Y)$, where $\phi_X$ and $\phi_Y$ are integers such that $\phi_X \ge 2Y$ and $\phi_Y \ge 2X$. Then, $\Gamma'$ is a grid drawing of $G$ in which the minimum distance between any vertex and the edge segment representing $(u,v)$ is at least~$\rho=\frac{1}{2}$.
\end{lemma}
\begin{proof}
Drawing $\Gamma'$ is a grid drawing of $G$, as it is obtained through an affine transformation of $\Gamma$ and both $\phi_X$ and $\phi_Y$ are integers. 
We prove that the minimum distance between any vertex and the edge segment representing $(u,v)$ is at least~$1$ (and thus at least $\rho=\frac{1}{2}$). To this aim, it suffices to consider the case in which $\phi_X = 2Y$ and $\phi_Y = 2X$, as for larger values the distance between $(u,v)$ and any vertex in $\Gamma'$ can only increase further. Up to a translation, we may assume that one endpoint of $(u,v)$ in $\Gamma$ is $(0,0)$, which implies that its other endpoint is $(X,Y)$. Since $\phi_X = 2Y$ and $\phi_Y = 2X$, the endpoints of $(u,v)$ in $\Gamma'$ are $(0,0)$ and $(2XY,2XY)$. Assume to the contrary that there is a vertex $w$ in $\Gamma'$, which is at a distance strictly less than~$1$ from $(u,v)$. It follows that $w$ must lie at a grid point either on line $l_w$ with slope $+1$ through the point $(1,0)$ or on line $l_w'$ with slope $+1$ through the point $(0,1)$. By symmetry, we may assume that the former situation applies. For some integer number $q$, let $(q+1,q)$ be the grid point representing $w$ along $l_w$ in $\Gamma'$. By the stretching factors $\phi_X$ and $\phi_Y$, the position of $w$ in $\Gamma$ is $(\frac{q+1}{2Y},\frac{q}{2X})$, which must be a grid point since $\Gamma$ is a grid drawing. Since both $2X$ and $2Y$ are even and either $q+1$ or $q$ is odd,  either $\frac{q+1}{2Y}$ or $\frac{q}{2X}$ is not integer, which contradicts the fact that $\Gamma$ is a grid drawing.
\end{proof}

\noindent The next theorem easily follows from the previous lemma.

\begin{theorem}[Stretching Theorem]\label{th:stretch}
Every graph that admits a $W \times H$ grid drawing also admits a \drawing on a grid of size $2HW \times 2HW$.
\end{theorem}
\begin{proof}
Let $\sigma_X$ and $\sigma_Y$ be the maximum $x$- and $y$-span over all edges in the~$W \times H$ grid drawing. Since $\sigma_X \leq W$ and $\sigma_Y \leq H$, the result follows by Lemma~\ref{le:stretch}.
\end{proof}

Corollary~\ref{co:kcolorable} is obtained by combining Theorem~\ref{th:stretch} and a result by Wood~\cite{DBLP:journals/comgeo/Wood05}, who proved that every $n$-vertex $k$-colorable graph has an $O(k) \times O(n)$ grid drawing.  Note that Corollary~\ref{co:kcolorable} applied to a planar graph yields a \drawing on a grid of quadratic size, which, however, is not necessarily planar. Corollary~\ref{co:planar} is an immediate implication of Theorem~\ref{th:stretch} and the fact that every $n$-vertex planar graph has an $O(n) \times O(n)$ grid drawing~\cite{DBLP:journals/combinatorica/FraysseixPP90,DBLP:conf/soda/Schnyder90}; a drastic improvement will be presented in Section~\ref{sec:planar}.

\begin{corollary}\label{co:kcolorable}
Every $k$-colorable graph with $n$ vertices admits a \drawing on a grid of size $O(k \, n) \times O(k \, n)$. 
\end{corollary}

\begin{corollary}\label{co:planar}
Every planar graph with $n$ vertices admits a planar \drawing on a grid of size $O(n^2) \times O(n^2)$.
\end{corollary}

\section{Nonplanar Drawings}
\label{sec:nonplanar}

In this section we study \drawings for two families of nonplanar graphs, namely bounded bandwidth graphs and complete graphs.

\subsection{Bounded bandwidth graphs}  A graph $G=(V,E)$ has \emph{bandwidth} $b$ if there is a total ordering of the vertices of $G$, denoted by $\prec_b$, such that for every edge $(u,v) \in E$ with $u \prec_b v$, the cardinality of the set  $\{w \in V: u \prec_b w \preceq_b v \}$ is at most $b$ (see, e.g.,~\cite{DBLP:journals/jcss/DubeyFU11,DBLP:conf/swat/Feige00}). We show that the graphs with bounded bandwidth admit \drawings in~linear~area.

\begin{theorem}\label{thm:bandwidth}
Every graph $G$ with $n$ vertices and bandwidth $b$ admits a \drawing on a grid of size $O(b\,n) \times O(b^2)$.
\end{theorem}
\begin{proof}
We assume that $G$ is maximal (i.e., no edge can be added without increasing its bandwidth). 
At a high level, for a prime number $p$, we first construct an $O(n) \times O(b)$ grid drawing $\Gamma$ of $G$ in which no three vertices forming a $3$-cycle in $G$ are collinear, and for any edge $(u,v)$ it holds that $\sigma_X(u,v) < p$ and $\sigma_Y(u,v) < p$. Applying Lemma~\ref{le:stretch} to $\Gamma$ yields the desired \drawing. 
To construct $\Gamma$, we make use of a result by Erd\H{o}s~\cite{erdos}, who showed that for every prime $p$, there do not exist three collinear points in the set consisting of the $p$ points
\begin{equation}\label{eq:b1}
p_i = ( i,i^2 \modulo{p}), \; i=0,1,\dots,p-1,
\end{equation}
\noindent which  follows from the fact that a quadratic polynomial can have at most two intersection points with a line over $\mathbb{F}_p$, i.e., over the prime field of order $p$.

Let  $v_0, v_1, \dots, v_{n-1}$ be the vertices of $G$ according to $\prec_b$. Let $p$ be the minimum prime number such that $p > b$. Our construction is the following simple extension of~\cite{erdos}; refer to Figure~\ref{fi:bandwidth} for an illustration. We map $v_i$ to
\begin{equation}\label{eq:b2}
p_i = ( i,i^2 \modulo{p}), \; i=0,1,\dots,n-1.
\end{equation}
\begin{figure}
\centering
\includegraphics{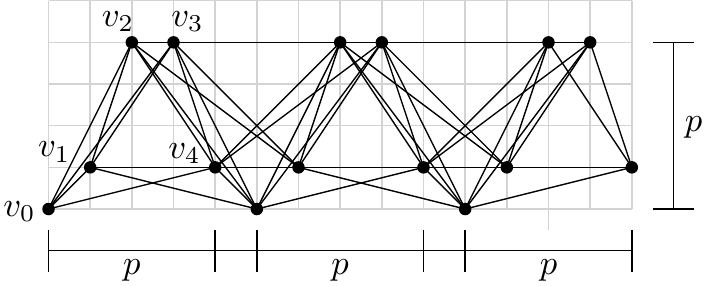}
\caption{Illustration for Theorem~\ref{thm:bandwidth}. A grid drawing of a graph with bandwidth $b=4$ computed by applying the described construction with $p=5$.  \label{fi:bandwidth}}
\end{figure}

\noindent  Note that the $y$-coordinates of the points defined in Eq.~(\ref{eq:b2}) have a period of $p$. Namely, when $i > p$,  we can assume that $i = t\, p + j$ for some $t \in \mathbb{N}$ and $j < p$. It follows that
\[i^2 \modulo{p} = (t\, p + j)^2 \modulo{p} = (t\, p)^2 + 2\,t\, p\, j + {j}^2 \modulo{p}= j^2 \modulo{p}.\]

We claim that no three vertices forming a $3$-cycle in $G$ are collinear in the constructed drawing. To see this, consider any three vertices $v_{k}$, $v_{\ell}$ and $v_{m}$ that form a $3$-cycle in $G$, and assume without loss of generality that $k < \ell < m$.  Since $G$ has bandwidth $b$, $m-k \le b<p$ holds, as otherwise $v_{k}$ and $v_{m}$ would not be adjacent. Hence, there is a set of $p$ consecutive vertices placed according to Eq.~(\ref{eq:b2}) containing $v_{k}$, $v_{\ell}$ and $v_{m}$. However, the proof by Erd\H{o}s \cite{erdos} ensures that such a set contains no three collinear points.

We constructed a grid drawing $\Gamma$ of $G$ on a grid of size $n \times p$. Additionally, for any edge $(u,v)$ it holds $\sigma_X(u,v) < p$ and $\sigma_Y(u,v) < p$. Then the result follows by Lemma~\ref{le:stretch} and by Bertrand's postulate, which shows~that~$p \le 2b$.  
\end{proof}

\subsection{Complete graphs} Corollary~\ref{co:kcolorable} implies that the complete graph $K_n$ admits a \drawing on a grid of size $O(n^2) \times O(n^2)$. We conclude this section by strengthening this result. Namely, the next theorem shows that the same area bound can be obtained by \drawings that are also convex. Here, a \emph{convex drawing} is a grid drawing where the vertices of the graph are placed at the corners of a convex polygon. We remark that, in contrast to Corollary~\ref{co:kcolorable}, the next theorem cannot be obtained by exploiting Theorem~\ref{th:stretch}. This is because of a known (super quadratic) lower bound on the area required to produce a convex grid  drawing of a complete graph, given by Acketa and Zunic~\cite{DBLP:journals/jct/AcketaZ95}. Furthermore, the presented bound cannot be improved asymptotically.

\begin{theorem}\label{thm:convex}
The complete graph $K_n$ with $n$ vertices  admits a convex \drawing on a grid of size $O(n^2) \times O(n^2)$ and this is sharp.
\end{theorem}
\begin{proof}
Denote by $v_0,v_1,\ldots,v_{n-1}$ the vertices of $K_n$. Let $R_n$ be a regular $n$-gon centered at point $(0,0)$ such that the distance between its center and any of its vertices is $r$, where $r$ is a positive integer that we will define below. For $i=0,1,\ldots,n-1$, we place vertex $v_i$ at the $i$-th corner of $R_n$ and obtain a drawing $\Gamma_n$ of $K_n$, which is not necessarily a grid drawing. For $i=0,1,\ldots,n-1$, denote by $x_i$ the distance between vertex $v_i$ and edge $(v_{i-1},v_{i+1})$, where the indices are taken modulo $n$. It follows that $x_0=x_1=\ldots=x_{n-1}$. Observe that the edge-vertex resolution of $\Gamma_n$ equals to $x_0$. The goal is to specify $r$  such that $x_0$ is at least $16$ (a suitable value greater than one). The next claim shows that if we set $r=2n^2$, this goal is achieved. 

\begin{claim}
If $r=2n^2$, then the distance $x_i$ between vertex $v_i$ and edge $(v_{i-1},v_{i+1})$ is at least $16$, where $n \geq 2$, $i=0,1,\ldots,n-1$ and indices taken modulo $n$.
\end{claim}
\noindent\textit{Proof of claim.} 
By the symmetry of the construction, it suffices to prove that $x_0 \geq 16$. For $i=0,1,\ldots,n-1$, denote by $\phi_i$ the smallest of the two angles between the line segments that connect the center of $R_n$ with the vertices $v_{i-1}$ and $v_{i}$. Since $R_n$ is a regular $n$-gon, it follows that $\phi_0=\phi_1=\ldots=\phi_{n-1}=\frac{2\pi}{n}$. Since the edge $(v_1,v_{n-1})$ is perpendicular to the line segment connecting the center $(0,0)$ of $R_n$ with vertex $v_0$, it follows that $\cos{(\frac{2\pi}{n})} = \frac{r-x_0}{r}$. Hence, the goal $x_0 \geq 16$ that we set above is equivalent to $r(1-\cos{ (\frac{2\pi}{n}) }) \geq 16$. 
Since $r = 2n^2$ and $n \geq 2$, what we have to prove is that $2n^2(1-\cos{\frac{2\pi}{n}}) \geq 16$ for every $n \geq 2$. To see this, let $f: \mathbb{R} \rightarrow \mathbb{R}$ be such that $f(x) = 2x^2 (1 - \cos{(\frac{2 \pi}{x})}) - 16$. Clearly, if $f(x) \geq 0$ for $x \geq 2$, then 
the proof follows. Using elementary properties of trigonometric functions, we can rewrite $f$ as $f(x) = 4 x^2 \sin^2{(\frac{\pi}{x})} - 16$. Since $x \geq 2$, $f(x) \geq 0$ is equivalent to $2 x \sin{(\frac{\pi}{x})} - 4 \geq 0$. Let $h: \mathbb{R} \rightarrow \mathbb{R}$ be such that $h(x) = 2 x \sin{(\frac{\pi}{x})} - 4$. The first derivative of $h$ is $h'(x) = 2 \sin{(\frac{\pi}{x})} - \frac{2\pi}{x} \cos{(\frac{\pi}{x})}$. Hence, $h'(x) \geq 0$ if and only if $\tan{(\frac{\pi}{x})} \geq \frac{\pi}{x}$, which holds for all $x \geq 2$. The fact the first derivative of $h$ is positive implies that $h$ is increasing. Hence, $h(x) \geq h(2)$ for all $x \geq 2$, or equivalently $2 x \sin{(\frac{\pi}{x})} - 4 \geq 0$ for $x \geq 2$. The latter implies that $f(x) \geq 0$ for $x \geq 2$.\qed


We now prove that the drawing $\Gamma_n'$ obtained from $\Gamma_n$ by rounding each vertex in $\Gamma_n$ to its nearest grid point in $\Gamma_n'$ has edge-vertex resolution at least $\frac{1}{2}$, that is, by replacing each vertex with a disk centered at that point we obtain a \drawing. Consider the effect of this rounding operation on the edge-vertex resolution of $\Gamma_n'$. In particular, consider vertex $v_i$ and the edge $(v_{i-1},v_{i+1})$ for some $i=0,1,\ldots,n-1$. The rounding may result in bringing edge $(v_{i-1},v_{i+1})$ one unit closer to $v_i$ in the worst case. Similarly, in the worst case the same effect may be observed by the rounding of the vertices $v_i$ and $v_{i+1}$. Hence, in the worst case the rounding may result in decreasing the edge-vertex resolution of $\Gamma_n$ by two units in $\Gamma_n'$.  This completes the proof of the upper~bound.

For the lower bound, we prove an even stronger statement. Consider any convex \drawing of the $n$-vertex cycle $C_n \subset K_n$. Without loss of generality, we assume that at least $1/8$ of its edges are in the `bottom-right' segment, i.e., their slope is between $0$ and $1$. As the slope increases along the bottom-right part as we move to the right, this means that we will find $n/32$ disjoint pairs of consecutive edges whose slopes differ by less than $32/n$; otherwise there would be $n/32$ pairs of consecutive edges whose slopes differ by at least $32/n$, summing to a total of 1, contradicting that all slopes are between 0 and 1.

Consider such a pair of edges and let the coordinates of their end-vertices be without loss of generality $(0,0)$, $(x_1,y_1)$ and $(x_2,y_2)$. Then the slope condition gives $y_2/x_2 - y_1/x_1 < 32/n$. Since each vertex is represented by a disk of radius $\rho$, we need that the link from $(0,0)$ to $(x_2,y_2)$ is at least $\rho$ above $(x_1,y_1)$ when it passes the line $x=x_1$. This means $\rho=\frac 12\le x_1(y_1+y_2)/(x_1+x_2)-y_1$. From here a simple calculation gives $x_1+x_2 \le 2x_1y_2-2x_2y_1$, but we also have $x_1y_2-x_2y_1<32x_1x_2/n$ from the condition that the slopes were close. Combining the two inequalities before, we get $x_1 + x_2 < 64 x_1 x_2/n$, which can be rewritten as $n/64 < x_1 x_2 / (x_1 + x_2)$, which in turn implies that $n/64 < x_1 + x_2$, since $x_1 x_2 / (x_1 + x_2) < x_1 < x_1 + x_2$. Summing up over all $n/32$ pairs of edges gives a quadratic lower bound for the width of the \drawing.
\end{proof}

\section{Planar Drawings}
\label{sec:planar}
In this section we study crossing-free \drawings of planar graphs. By Corollary~\ref{co:planar} every planar graph admits a planar \drawing on a grid of quartic size; we reduce this upper bound to quadratic, which is tight even for planar grid drawings~\cite{DBLP:conf/stoc/FraysseixPP88}. Moreover, we prove a quadratic lower bound for the area requirement of star graphs, which notably holds also for closed RI drawings. We begin with this last result.

\subsection{Lower bound on the area of star graphs}
In the traditional straight-line drawing model, an $n$-vertex star admits a planar drawing on a grid of size $2 \times (n-2)$, e.g., by placing its center at $(0,0)$, its $i$-th leaf at $(i-1,1)$, where $i=1,\ldots,n-2$, and its $(n-1)$-th leaf at $(1,0)$. We prove a quadratic lower bound for the area of \drawings of stars. 

\newcommand{\nplus}[1]{^{^{\hspace{-8pt}\small #1}}}

\begin{theorem}\label{thm:star}
Any \drawing of the $n$-vertex star requires a grid of size $\Omega(n^2)$.
\end{theorem}
\begin{proof}
Let $\Gamma$ be any \drawing of the $n$-vertex star $S_n$, and denote by $c$ the vertex of $S_n$ with degree $n-1$. We will use the following notation; refer to Figure~\ref{fi:star} for an illustration. The leaves of $S_n$ (i.e., its degree-1 vertices) are denoted by $v_0, v_1, \dots, v_{n^*}$, with $n^*=n-2$, such that $(c,v_0), (c,v_1), \dots, (c,v_{n^*})$ is a counter-clockwise order of the edges around $c$. For every leaf $v_i$ of $S_n$, the center of the disk representing $v_i$ in $\Gamma$ is denoted by $p_i$. Moreover, up to a translation of $\Gamma$, we can assume without loss of generality that the center $p_c$ of the disk representing $c$ is at point $(0,0)$, namely $p_c=(0,0)$. For $0 \leq i \leq n^*$, we denote by $\theta_i$ the angle formed by the line segments $\overline{p_cp_i}$ and $\overline{p_cp_{i+1}}$ (indices taken modulo $n^*+1$). We assume that each such angle is strictly smaller than $\pi$, as otherwise we can draw one or two extra leaves in $\Gamma$ to guarantee this property without affecting the asymptotic area requirement with respect to $n$. Namely, if there exists an angle $\theta_i = \pi$, there is one grid point at distance one from $c$ such that drawing a new leaf on this point splits $\theta_i$ into two smaller angles; if $\theta_i > \pi$, there are two grid points at distance one from $c$ such that drawing a new leaf on each of them splits $\theta_i$ into three angles smaller than $\pi$.
Clearly $\sum_{i=0}^{n^*}\theta_i=2 \pi$. Finally, we denote by $T_i$ the triangle having $p_c$, $p_i$, and $p_{i+1}$ as corners, 
and by $A(T_i)$ the area of triangle $T_i$. Since $\rho=\frac{1}{2}$, the height of $T_i$ is at least $\frac{1}{2}$ (see, for example, triangle $T_1$ in Figure~\ref{fi:star}),  it follows that:
\begin{equation}\label{eq:area}
A(T_i) \quad \ge \quad \frac{|\overline{p_cp_i}|}{4}.
\end{equation}
\noindent On the other hand, using elementary trigonometry, one can see that 
\begin{equation}\label{eq:sin}
\sin \theta_i \quad \ge \quad \frac{1}{2|\overline{p_cp_i}|}.
\end{equation}
%
\begin{figure}[t]
\centering
\includegraphics[page=3]{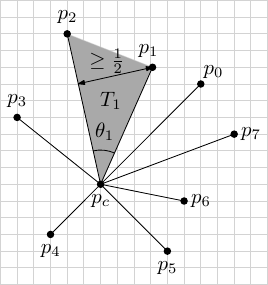}
\caption{Illustration for the proof of Theorem~\ref{thm:star}.\label{fi:star}}
\end{figure}
\noindent Since $\sin{x}<x$ when $x \in (0,\pi)$, it holds
\begin{equation}\label{eq:upper-sum}
\sum_{i=0}^{n^*} \sin{\theta_i}
\quad < \quad 
\sum_{i=0}^{n^*} \theta_i \quad = \quad 2\pi.
\end{equation}
%
%
\noindent We are now ready to put everything together:
\begin{equation}\label{eq:upper-const}
\sum_{i=0}^{n^*} \frac{1}{A(T_i)} 
\quad \leq\nplus{(\ref{eq:area})} \quad
\sum_{i=0}^{n^*} \frac{4}{|\overline{p_cp_i}|} 
\quad \leq\nplus{(\ref{eq:sin})} \quad \sum_{i=0}^{n^*} 8 \sin{\theta_i} 
\quad <\nplus{(\ref{eq:upper-sum})} \quad
 16\pi.
\end{equation}
\noindent Finally, by the arithmetic-harmonic mean inequality and since $n^*=n-2$ we have that:
\begin{equation*}\label{eq:final-area}
\sum_{i=0}^{n^*} \frac{A(T_i)}{n^*+1} \quad \geq \quad \frac{n^*+1}{\sum_{i=0}^{n^*}\frac{1}{A(T_i)}} \quad \geq\nplus{(\ref{eq:upper-const})} \quad \frac{n^*+1}{16\pi} \quad \Rightarrow \quad \sum_{i=0}^{n^*} A(T_i) \quad \geq \quad \frac{(n-1)^2}{16\pi}.
\end{equation*}
Since the area of $\Gamma$ is at least the sum of the areas of all triangles $T_i$, the statement~follows.
\end{proof}

The next corollary follows from Theorem~\ref{thm:star}, as any closed RI drawing can be interpreted as a \drawing. Note that previously, the most restricted family of planar graphs for which a quadratic area lower bound was known was the irreducible triangulations~\cite{DBLP:journals/comgeo/SadasivamZ11}.

\begin{corollary}\label{co:star}
Any closed RI drawing of the $n$-vertex star requires a grid of size $\Omega(n^2)$.
\end{corollary}

\subsection{Planar graphs} 
We assume familiarity with basic concepts of planar graph drawing~\cite{DBLP:reference/crc/Vismara13}. 
For completeness, we first recall a standard graph drawing algorithm by de~Fraysseix, Pach and Pollack~\cite{DBLP:journals/combinatorica/FraysseixPP90}, called \emph{shift-method} (for a textbook-level description of the shift method see~\cite[Section 4.2]{DBLP:books/ws/NishizekiR04}).
The algorithm builds upon the well-known \emph{canonical ordering} for maximal planar graphs~\cite{DBLP:journals/combinatorica/FraysseixPP90}, which is defined as follows. Let $G=(V,E)$ be a maximal planar graph and let $\pi =(v_1,\ldots,v_n)$ be a permutation of $V$. Assume that edges $(v_1,v_2)$, $(v_2,v_n)$ and $(v_1,v_n)$ form a face of $G$, which we assume without loss of generality to be its outer face. For $k=1,\ldots,n$, let $G_k$ be the subgraph induced by $\bigcup_{i=1}^k \{v_i\}$ and denote by $C_k$  the outer face of $G_k$. Then, $\pi$ is a \emph{canonical ordering} of $G$ if for each $k=3,\ldots,n$ the following properties hold: %
\begin{enumerate}[P.1]
\item $G_k$ is biconnected,
\item all neighbors of $v_k$ in $G_{k-1}$ are consecutive on $C_{k-1}$, and
\item if $k \neq n$, then $v_k$ has at least one neighbor $v_j$, with $j > k$.
\end{enumerate}
A canonical ordering of a maximal planar graph always exists and can be computed in $O(n)$ time~\cite{DBLP:conf/stoc/FraysseixPP88}.

\begin{figure}[t!]
	\centering
	\begin{subfigure}{.48\textwidth}
	\flushleft
	\includegraphics[page=1,width=\textwidth]{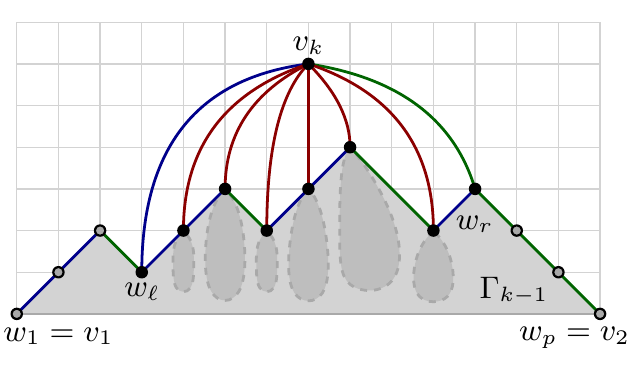}
	\subcaption{Contour condition}
	\label{fig:shift-method-1}
	\end{subfigure}
	\hfil
	\begin{subfigure}{.48\textwidth}
	\flushright
	\includegraphics[page=2,width=\textwidth]{figs/canonical2}
	\subcaption{Placement of $v_k$ in $\Gamma_{k-1}$}
	\label{fig:shift-method-2}
	\end{subfigure}
	\caption{%
	Illustration of the shift-method by de~Fraysseix, Pach and Pollack~\cite{DBLP:journals/combinatorica/FraysseixPP90}.}
	\label{fig:shift-method}
\end{figure}

The \emph{shift-method}~\cite{DBLP:journals/combinatorica/FraysseixPP90} is an incremental algorithm, which constructs a planar drawing $\Gamma$ of a maximal planar graph $G=(V,E)$; in the following, we refer to the linear-time variant by Chrobak and Payne~\cite{DBLP:journals/ipl/ChrobakP95}. Drawing $\Gamma$ has integer grid coordinates and fits in a grid of size $(2n-4)\times (n-2)$. More precisely, based on a canonical ordering $\pi$ of $G$, drawing $\Gamma$ is constructed as follows. Initially, vertices $v_1$, $v_2$ and $v_3$ are placed at points $(0,0)$, $(2,0)$ and $(1,1)$, respectively. For $k=4,\ldots,n$, assume that a planar grid drawing $\Gamma_{k-1}$ of $G_{k-1}$ has been constructed in which edges of $C_{k-1}$ are drawn as  straight-line segments with slopes $\pm 1$, except for the edge $(v_1,v_2)$, which is drawn as a horizontal line segment (\emph{contour condition}; see Figure~\ref{fig:shift-method-1}). Also, for $i=1,\ldots,k-1$ vertex $v_i$ has been associated with a so-called \emph{shift-set} $S(v_i)$. For $v_1$, $v_2$ and $v_3$, it holds that $S(v_1)=\{v_1\}$, $S(v_2)=\{v_2\}$ and  $S(v_3)=\{v_3\}$. Let $(w_1,\ldots,w_p)$ be the vertices of $C_{k-1}$ from left to right in $\Gamma_{k-1}$, where $w_1=v_1$ and $w_p=v_2$. Let also $(w_\ell,\ldots,w_r)$, with $1 \leq \ell < r \leq p$ be the neighbors of $v_k$ from left to right along $C_{k-1}$ in $\Gamma_{k-1}$. To avoid edge-overlaps, the algorithm first translates each vertex in $\bigcup_{i=\ell+1}^{r-1} S(w_i)$ one unit to the right and each vertex in $\bigcup_{i=r}^{p} S(w_i)$ two units to the right; see Figure~\ref{fig:shift-method-2}. Then, the algorithm places vertex $v_k$ at the intersection of the line of slope $+1$ through $w_\ell$ with the line of slope $-1$ through $w_r$ (which is a grid point, since by the contour condition the Manhattan distance between $w_\ell$ and $w_r$ is even) and sets $S(v_k)=\{v_k\} \cup \bigcup_{i=\ell+1}^{r-1}S(w_i)$.

While constructing drawing $\Gamma$, it is also possible to compute a $3$-coloring of the edges of $G$, which is known as \emph{Schnyder realizer} in the literature~\cite{felsner,DBLP:conf/soda/Schnyder90}. In particular, color $(v_1,v_3)$ blue, $(v_2,v_3)$ green and when a vertex $v_k$ with $k=4,\ldots,n$ is placed, color edge $(w_\ell,v_k)$ blue, edge $(v_k,w_r)$ green and the remaining edges incident to $v_k$ in $G_k$ red, that is, $(w_i,v_k)$ with $i=\ell+1,\ldots,r-1$. It follows that all edges that appear in the contour of $\Gamma_k$ are either blue or green, which further implies that all faces of $\Gamma_k$ (and thus of $\Gamma$) are either bichromatic or trichromatic. Since vertices in the same shift-set are always translated by the same amount, the red edges are \emph{rigid}, i.e., neither the slope nor the length of a red edge incident to $v_k$ in $G_k$ can change due to a shift required by the placement of a vertex $v_h$ with $k < h \leq n$. Consider now an edge $e$ in $\Gamma$ and let $\phi(e)$ be the slope of $e$. The construction ensures that if $e$ is blue, then $0 < \phi(e) \leq \pi/4$; if $e$ is green, then $3\pi/4 \leq \phi(e) < \pi$; if $e$ is red, then $\pi/4 < \phi(e) < 3\pi/4$.

\medskip

\begin{figure}[t!]
	\centering
	\includegraphics[page=3,scale=0.9]{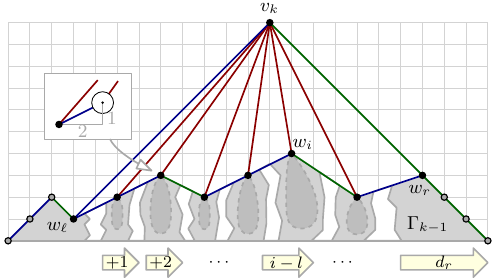}	\caption{%
	Illustration of the placement of $v_k$ in $\Gamma_{k-1}$ in the modification of the shift-method.}	
	\label{fig:shift-method-3}
\end{figure}

Chrobak, Goodrich and Tamassia~\cite{DBLP:conf/compgeom/ChrobakGT96}, in order to compute quadratic-area grid drawings of planar graphs with constant edge-vertex resolution, suggested that ``at the time when a new vertex is installed we shift all covered vertices to the right, ensuring that they are far from nonincident edges''. The main challenge is to determine which vertices to shift and by how much, while keeping both the area small and the edge-vertex resolution constant. Here we provide a linear-time algorithm that addresses this challenge. We start by placing $v_1$, $v_2$ and $v_3$ as in the original shift-method. For placing $v_k$, with $k=4,\ldots,n$, our algorithm shifts the vertices of $\Gamma_{k-1}$ in three ``shifting phases''. First, each vertex in $S(w_i)$ with $i=\ell+1,\ldots,r-1$ is shifted by $i-\ell$ units to the right (instead of a single unit, as in the original shift-method). In the second phase, each vertex in $S(w_r)$ is shifted by $d_r$ units to the right, where $d_r$ is  either $r-\ell$ or $r-\ell+1$ so to guarantee that the Manhattan distance between $w_\ell$ and $w_r$ is even. In the final phase, each vertex in $\bigcup_{i=r+1}^{p} S(w_i)$ is moved by $d_r$ units to the right\footnote{Note that although the second and the third shifting phases shift the relevant vertices by the same amount $d_r$, we distinguish the two phases for clarity of presentation.}; see Figure~\ref{fig:shift-method-3}. After all three shifting phases have been executed, we have the final placement for the vertices of $G_{k-1}$ in $\Gamma_k$. We complete the construction of $\Gamma_k$ by placing vertex $v_k$ at the intersection of the line of slope $+1$ through $w_\ell$ with the line of slope $-1$ through $w_r$, as in the original shift-method. Hence, the contour condition  is maintained, assuming that the coordinates of $v_k$ are integer (a property which is formally proven in the following).

Observe that the first shifting phase implies that the horizontal distance between any two consecutive vertices $w_i$ and $w_{i+1}$ in $C_{k-1}$ with $i \in \{\ell, \ldots, r-2 \}$ gets increased by one unit in $\Gamma_k$, while in the original shift-method this would only be the case for $w_\ell$ and $w_{\ell+1}$. In the second shifting phase, the choice of $d_r$ guarantees that if $v_k$ is placed at the intersection of the line of slope $+1$ through $w_\ell$ with the line of slope $-1$ through $w_r$, then its position coincides with a grid point. This is due to the fact that an even Manhattan distance between $w_\ell$ and $w_r$ implies that the two aforementioned lines intersect at a grid point~\cite{DBLP:journals/combinatorica/FraysseixPP90}. The choice of $d_r$ further implies that the horizontal distance between $w_{r-1}$ and $w_r$ gets increased by either one or two units in $\Gamma_k$, while in the original shift-method the corresponding increment is always one unit. Finally, notice that the third translation phase does not affect the horizontal distances of the involved vertices that are on $C_{k-1}$, as in the original shift-method.

Since the contour condition is maintained in the course of the construction, the planarity of $\Gamma_k$ is implied as in the original shift-method. To complete the proof of correctness of our algorithm, we first prove in the following lemma that $\Gamma_k$ is a \drawing of $G_k$.

\begin{lemma}\label{lem:properties}
Let $\Gamma_{k-1}$ be a \drawing of $G_{k-1}$ computed by our algorithm. 
Then the following two properties hold: (i)~the drawing produced by applying the three shifting phases on the vertices of $\Gamma_{k-1}$ has edge-vertex resolution at least $1/2$, and (ii)~the newly introduced edges of $\Gamma_k$ incident to $v_k$ leave the edge-vertex resolution of $\Gamma_k$ at least $1/2$.
\end{lemma}
\begin{proof}
To prove~(i), consider any triangular face $f$ in $\Gamma_{k-1}$, delimited by $u$, $v$ and $w$. If face $f$ is not stretched in some of the three shifting phases, then clearly there is no edge-disk intersection in the drawing of $f$ in $\Gamma_{k}$.  Hence, we may assume that $f$ has been stretched in at least one of the three shifting phases. As already mentioned, $f$ is either bichromatic or trichromatic in the Schnyder realizer. We consider these two cases separately. 

\begin{figure}[t!]
	\centering
	\begin{subfigure}{.24\textwidth}
	\includegraphics[page=1,width=\textwidth]{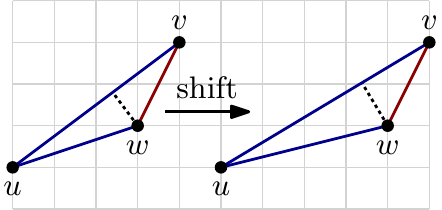}
	\subcaption{Case A.1.}
	\label{fig:shift-bbr}
	\end{subfigure}
	\hfil
	\begin{subfigure}{.25\textwidth}
	\includegraphics[page=2,width=\textwidth]{figs/shift}
	\subcaption{Case A.2. }
	\label{fig:shift-bbg}
	\end{subfigure}
	\hfil
	\begin{subfigure}{.24\textwidth}
	\includegraphics[page=3,width=\textwidth]{figs/shift}
	\subcaption{Case B.1.}
	\label{fig:shift-brg1}
	\end{subfigure}
	\hfil
	\begin{subfigure}{.24\textwidth}
	\includegraphics[page=4,width=\textwidth]{figs/shift}
	\subcaption{Case B.2.}
	\label{fig:shift-brg2}
	\end{subfigure}
	\caption{%
	Stretching of a triangular face $f = \langle u,v,w \rangle$ during the shifting phases.}
	\label{fig:shift-cases}
\end{figure}

\begin{itemize}

\item \textbf{Case A:} $f$ is bichromatic. Since the red edges are rigid, i.e., their length stay unchanged in the course of the algorithm, the color appearing twice in $f$ cannot be red. So we shall assume that $f$ has two blue edges; the case in which it has two green edges is symmetric. Let, without loss of generality,  $(u,v)$ and $(u,w)$ be the  blue edges of $f$, and thus $(v,w)$ is either red or green. 

\begin{itemize}

\item \textbf{Case A.1:} Assume first that $(v,w)$ is red; refer to Figure~\ref{fig:shift-bbr}.  Since  red edges are rigid and we assumed that $f$ is stretched, it follows $v$ and $w$ are in the same shift-set, while $u$ is in a different shift-set from the one of $v$ and $w$. Since $(u,v)$ and $(u,w)$ are blue, we have $\phi(u,v),\phi(u,w) \in (0,\pi/4]$, which implies that $u$ is to the left and below both $v$ and $w$. The fact that $(v,w)$ is red implies that $\phi(v,w) \in (\pi/4,3\pi/4)$. Without loss of generality, we assume that $w$ is below $v$, as in Figure~\ref{fig:shift-bbr}, while $w$ can be either to the left or to the right of $v$. Since $u$ is to the left of $v$ and $w$, $u$ is shifted by a smaller amount than $v$ and $w$. Since the angle $\angle (uvw)$ is increased by the shift while the length of $(v,w)$ remains unchanged, the distance of $w$ to $(u,v)$ is increased after the shifting. Similarly, since $v$ is at least one unit above $w$ and since $w$ is at least one unit above $u$, the distance between $v$ and $(u,w)$ cannot be less than $1$. Thus, the edge-vertex resolution is at least $1/2$, as desired. 

\item \textbf{Case A.2:}  Assume now that $(v,w)$ is green; refer to Figure~\ref{fig:shift-bbg}. Assume again without loss of generality $v$ is above $w$, and recall that $u$ is below both $v$ and $w$. Moreover, $v$ appears between $u$ and $w$ in the horizontal direction. Consider now vertex $v$. Similarly as before, the lowest point of a disk with radius $1/2$ centered at $v$ has $y$-coordinate greater than the $y$-coordinate of $w$, and thus it cannot intersect edge $(u,w)$. Similar arguments can be made about vertices $u$ and $w$ and their opposite edges, respectively, which completes the case in which $f$ is bichromatic.

\end{itemize}

\item \textbf{Case B:}  $f$ is trichromatic. Without loss of generality let $v$ be the topmost vertex of $f$ in $\Gamma_{k-1}$ and let $u$ and $w$ appear in this order in a counterclockwise traversal of $f$ starting from $v$. This implies that $(u,v)$ is either blue or red, as if it was green then either $(u,w)$ or $(v,w)$ had to be also green, which is not possible since $f$ is trichromatic. We distinguish the two subcases below.

\begin{itemize}

\item \textbf{Case B.1:} Assume first $(u,v)$ is blue; refer to Figure~\ref{fig:shift-brg1}. Consider now vertex $w$. Since $(u,w)$ is green and $u$ is to the left and above $w$, the highest point of a disk with radius $1/2$ centered at $w$ has $y$-coordinate smaller than the $y$-coordinate of $u$, and thus it cannot intersect edge $(u,v)$. Concerning vertex $v$, note that the topmost endpoint of $(u,w)$ is at least one unit below it, and hence the distance between $v$ and $(u,w)$ is at least $1/2$.  Finally, the shift only increases the distance between vertex $u$ and edge $(v,w)$. 

\item \textbf{Case B.2:} Assume now $(u,v)$ is red; refer to Figure~\ref{fig:shift-brg2}. Consider now vertex $u$. Since $(u,w)$ is blue and $u$ is to the left and below $w$, the highest point of a disk with radius $1/2$ centered at $u$ has $y$-coordinate smaller than the $y$-coordinate of $w$, and thus it cannot intersect edge $(w,v)$. Concerning vertex $v$, note that the topmost endpoint of $(u,w)$ is at least one unit below it, and hence the distance between $v$ and $(u,w)$ is at least $1/2$.  Finally, the shift only increases the distance between vertex $w$ and edge $(u,w)$. 

\end{itemize}
\end{itemize}

To prove (ii), recall that the first shifting phase implies that the horizontal distance between any two consecutive vertices $w_i$ and $w_{i+1}$ in $C_{k-1}$ with $i \in \{\ell, \ldots, r-2 \}$ increases by one unit in $\Gamma_k$. Also, the second shifting phase guarantees that the horizontal distance between $w_{r-1}$ and $w_r$ increases by either one or two units in $\Gamma_k$. These, together with the fact that the absolute value of the slope of any edge $(w_i,v_k)$ for $i = \ell,\ldots, r$ is at least one, guarantee that the edge-vertex resolution is at least $\sqrt{2}/2 > 1/2$. 
\end{proof}

We are now ready to prove our main theorem.

\begin{theorem}\label{thm:planar}
Every planar graph $G$ with $n$ vertices admits a planar \drawing on a grid of size $(3n-7) \times \lceil (3n-7)/2 \rceil$.
\end{theorem}
\begin{proof}
After possibly augmenting $G$ with edges to make it maximal, let $\Gamma$ be a planar grid drawing of $G$ computed by our algorithm. By Lemma~\ref{lem:properties}, $\Gamma$ is in fact a \drawing of $G$. Thus, it remains to estimate the area required by $\Gamma$. We make use of an important property of Schnyder realizers, namely, that each monochromatic subgraph of $G\setminus\{(v_1,v_2),(v_1,v_n),(v_2,v_n)\}$ induces a tree with $n-2$ vertices~\cite{DBLP:conf/soda/Schnyder90}. By the contour condition, $\Gamma$ is contained in an isosceles right triangle. Hence, to determine its area, it is enough to determine its width. 
Our modification of the shift-method elongates some edges, which were not elongated by the original method. In particular, when placing $v_k$, the edges in the path from $w_{\ell+1}$ to $w_{r-1}$ in $C_{k-1}$ are horizontally stretched by exactly one unit. Since after the placement of $v_k$ these edges connect vertices in shift-set $S(v_k)$, they are not further elongated, that is, they are elongated exactly once in the course of the algorithm. Furthermore, in the original shift-method the edge $(w_{r-1},w_r)$ is elongated by one unit in the horizontal direction during the placement of $v_k$, while in our construction it might be necessary to be elongated by an extra unit.  To estimate the width of $\Gamma$ it is enough to estimate the additional width that is due to our modified shift-method. Towards this, we observe that we can charge the elongation of each edge $(w_i,w_{i+1})$ to the red edge $(w_i,v_k)$ for $i=\ell+1,\dots,r-1$. Hence, the additional width that is due to our modified shift-method is at most $n-3$, since the red subgraph of $G$ is a tree with exactly $n-3$ edges. Given that the width of the drawings produced by the original shift-method is at most $2n-4$, it follows that the width of the drawings of our algorithm is at most $(2n-4)+(n-3)=3n-7$. Since $\Gamma$ is contained in an isosceles right triangle (by the contour condition), its height is  $\lceil (3n-7)/2 \rceil$.
\end{proof}


We conclude with the pseudocode for a linear-time implementation of the algorithm
supporting Theorem~\ref{thm:planar}, see Figure~\ref{fig:algoritm}. The pseudocode is based on the linear-time implementation
of the shift-method by Chrobak and Payne~\cite{DBLP:journals/ipl/ChrobakP95}.

The shift-method can easily be implemented to run in quadratic time by updating the coordinates of all vertices
contained in the shift-sets explicitly at every vertex addition.  
In the original work of de~Fraysseix, Pach and Pollack~\cite{DBLP:journals/combinatorica/FraysseixPP90} a rather involved approach is used to achieve
a runtime of $O(n\log n)$.    
Later, Chrobak and Payne described a linear-time algorithm whose 
key ingredient is to store only relative $x$-coordinates rather than absolute values. 
This method required them to change the definition of the shift-set. 
The proof of Theorem~\ref{thm:planar} uses their definition of shift-set which enables us to 
adapt their approach for our needs.

In the pseudocode we use, besides the already introduced notation, some more variables. 
For a vertex $w_i$ that is part of the contour, $d(w_i)$ denotes the horizontal distance to its predecessor.
Furthermore, the shift-sets are stored as a forest of trees induced by the red edges. For every vertex
we store its link to the parent in the corresponding variable. The relative horizontal offset of a vertex $v$ 
to its parent is denoted by $\Delta(v)$. Note that the \textsc{GetLeftRight($v_k$, $C_{k-1}$)} routine, which returns the leftmost and rightmost neighbors of $v_k$ along $C_{k-1}$, can be easily implemented to run in time linear in the degree of $v_k$ and hence in overall linear time. 

\section{Discussion and Future Research Directions}
\label{sec:conclusions}

We remark that our results are all proved via constructive techniques, and it is possible to show that each of them can be implemented to run in linear time in the number of edges of the graph. The only exception is Theorem~\ref{thm:bandwidth}, which requires a linear ordering of the vertices with minimum bandwidth. Determining the bandwidth of a graph is NP-hard~\cite{DBLP:books/fm/GareyJ79}, even to approximate within a constant factor~\cite{DBLP:journals/jcss/DubeyFU11}; nonetheless there are classes of graphs~for which the problem becomes tractable or it can be approximated (see~\cite{DBLP:journals/jcss/DubeyFU11,DBLP:conf/swat/Feige00} for references). 
Our research suggests interesting research directions, among them:
\begin{enumerate}[P.1]
\item Establishing improved area bounds for specific subclasses of planar graphs;.
\item Designing trade-offs between the edge-vertex resolution of a \drawing and its area requirement.
\item Extending the proposed model by allowing (at most) one bend per edge.
\end{enumerate}

\begin{figure}[p]
\includegraphics[page=2]{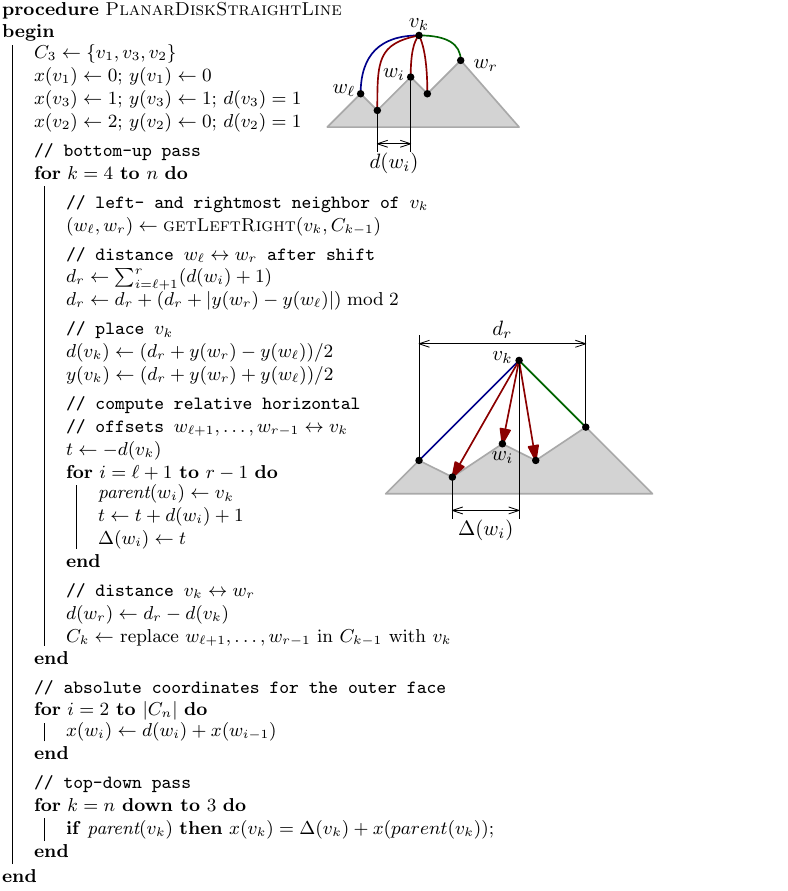}
\caption{A linear-time implementation of the algorithm described in the proof of Theorem~\ref{thm:planar}.}
\label{fig:algoritm}
\end{figure}

\paragraph*{Acknowledgments} We thank the anonymous reviewers for their useful comments and suggestions.

\clearpage 
\bibliographystyle{plainurl}
\bibliography{edgevertex}
\end{document}